\documentclass[twoside,leqno,twocolumn]{article}  
\usepackage{graphicx}
\usepackage{amssymb}
\usepackage{mathrsfs}
\usepackage{amsthm}
\usepackage[boxed,section]{algorithm}
\usepackage{algorithmicx}
\usepackage{algpseudocode}
\usepackage{amsmath}
\usepackage{subfig}
\usepackage{url}
\usepackage[vmargin=1in, hmargin=0.77in]{geometry}

\newtheorem{theorem}{Theorem}[section]
\newtheorem{definition}{Definition}
\newtheorem{lemma}{Lemma}[section]
\newtheorem{corollary}{Corollary}[section]

\newcommand{\abs}[1]{\left|#1\right|}

\newcommand{\norm}[2]{\left \|#2\right \|_{#1}}
\newcommand{\err}[3]{\mathrm{Err}_{#1}(#3, #2)}

\newcommand{\smbox}[1]{\mbox{\scriptsize #1}}

\DeclareMathOperator*{\median}{median}
\DeclareMathOperator{\supp}{supp}
\DeclareMathOperator{\E}{E}

\def\G{\mathbb{G}}
\def\R{\mathbb{R}}
\def\Z{\mathbb{Z}}
\def\eps{\epsilon}

\begin{document}

\title{\Large Efficient Sketches for the Set Query Problem\thanks{This
    research has been supported in part by the David and Lucille
    Packard Fellowship, MADALGO (Center for Massive Data Algorithmics,
    funded by the Danish National Research Association), NSF grant
    CCF-0728645, a Cisco Fellowship, and the NSF Graduate Research
    Fellowship Program.}}
\author{Eric Price\thanks{MIT CSAIL}}
\date{}

\maketitle

%\pagenumbering{arabic}
%\setcounter{page}{1}%Leave this line commented out.

\begin{abstract} \small\baselineskip=9pt
  We develop an algorithm for estimating the values of a vector $x \in
  \R^n$ over a support $S$ of size $k$ from a randomized sparse binary
  linear sketch $Ax$ of size $O(k)$.  Given $Ax$ and $S$, we can
  recover $x'$ with $\norm{2}{x' - x_S} \leq \eps \norm{2}{x - x_S}$
  with probability at least $1 - k^{-\Omega(1)}$.  The recovery takes
  $O(k)$ time.

  While interesting in its own right, this primitive also has a number
  of applications.  For example, we can:

  \begin{enumerate}
  \item Improve the linear $k$-sparse recovery of heavy hitters in
    Zipfian distributions with $O(k \log n)$ space from a $1+\eps$
    approximation to a $1 + o(1)$ approximation, giving the first such
    approximation in $O(k \log n)$ space when $k \leq O(n^{1-\eps})$.
  \item Recover block-sparse vectors with $O(k)$ space and a $1+\eps$
    approximation.  Previous algorithms required either $\omega(k)$ space
    or $\omega(1)$ approximation.
  \end{enumerate}

\end{abstract}

\section{Introduction}\label{sec:introduction}

In recent years, a new ``linear'' approach for obtaining a succinct
approximate representation of $n$-dimensional vectors (or signals) has
been discovered.  For any signal $x$, the representation is equal to
$Ax$, where $A$ is an $m \times n$ matrix, or possibly a random
variable chosen from some distribution over such matrices.  The vector
$Ax$ is often referred to as the {\em measurement vector} or {\em
  linear sketch} of $x$.  Although $m$ is typically much smaller than
$n$, the sketch $Ax$ often contains plenty of useful information about
the signal $x$.

A particularly useful and well-studied problem is that of {\em stable
  sparse recovery}.  The problem is typically defined as follows: for
some norm parameters $p$ and $q$ and an approximation factor $C>0$,
given $Ax$, recover a vector $x'$ such that
\begin{align}
\label{e:lplq}
\norm{p}{x' - x} &\le C \cdot \err{q}{k}{x},\\
\notag\mbox{\ where\ } \err{q}{k}{x} &=  \min_{k\mbox{-sparse } \hat{x}}  \norm{q}{\hat{x} - x}
\end{align}
where we say that $\hat{x}$ is $k$-sparse if it has at most $k$
non-zero coordinates.  Sparse recovery has applications to numerous
areas such as data stream computing~\cite{Muthu:survey, I-SSS} and
compressed sensing~\cite{CRT06:Stable-Signal,
  Don06:Compressed-Sensing}, notably for constructing imaging systems
that acquire images directly in compressed
form~(e.g.,~\cite{DDTLTKB,Rom}).  The problem has been a subject of
extensive study over the last several years, with the goal of
designing schemes that enjoy good ``compression rate'' (i.e., low
values of $m$) as well as good algorithmic properties (i.e., low
encoding and recovery times).  It is known that there exist
distributions of matrices $A$ and associated recovery algorithms that
for any $x$ with high probability produce approximations $x'$
satisfying Equation~\eqref{e:lplq} with $\ell_p=\ell_q=\ell_2$,
constant approximation factor $C=1+\eps$, and sketch length $m=O(k \log
(n/k))$;\footnote{In particular, a random Gaussian matrix~\cite{CD04}
  or a random sparse binary matrix~(\cite{glps}, building on
  \cite{ccf,CM03b}) has this property with overwhelming probability.
  See~\cite{GI} for an overview.} it is also known that this sketch
length is asymptotically optimal~\cite{dipw,FPRU}.  Similar results
for other combinations of $\ell_p$/$\ell_q$ norms are known as well.

Because it is impossible to improve on the sketch size in the general
sparse recovery problem, recently there has been a large body of work
on more restricted problems that are amenable to more efficient
solutions.  This includes \emph{model-based compressive
  sensing}~\cite{BCDH}, which imposes additional constraints (or
\emph{models}) on $x$ beyond near-sparsity.  Examples of models
include \emph{block sparsity}, where the large coefficients tend to
cluster together in blocks~\cite{BCDH,EKB}; \emph{tree sparsity},
where the large coefficients form a rooted, connected tree
structure~\cite{BCDH,LD}; and being \emph{Zipfian}, where we require
that the histogram of coefficient size follow a \emph{Zipfian} (or
\emph{power law}) distribution.

A sparse recovery algorithm needs to perform two tasks: locating the
large coefficients of $x$ and estimating their value.  Existing
algorithms perform both tasks at the same time.  In contrast, we
propose decoupling these tasks.  In models of interest, including
Zipfian signals and block-sparse signals, existing techniques can
locate the large coefficients more efficiently or accurately than they
can estimate them.  Prior to this work, however, estimating the large
coefficients after finding them had no better solution than the
general sparse recovery problem.  We fill this gap by giving an
optimal method for estimating the values of the large coefficients
after locating them.  We refer to this task as the \emph{Set Query
  Problem}\footnote{The term ``set query'' is in contrast to ``point
  query,'' used in e.g.~\cite{CM03b} for estimation of a single coordinate.}.

\textbf{Main result.}  (Set Query Algorithm.) We give a randomized
distribution over $O(k) \times n$ binary matrices $A$ such that, for
any vector $x \in \R^n$ and set $S \subseteq \{1, \dotsc, n\}$ with
$\abs{S} = k$, we can recover an $x'$ from $Ax + \nu$ and $S$ with
\[
\norm{2}{x' - x_S} \leq \eps ( \norm{2}{x - x_S} + \norm{2}{\nu})
\]
where $x_S \in \R^n$ equals $x$ over $S$ and zero elsewhere.  The
matrix $A$ has $O(1)$ non-zero entries per column, recovery succeeds
with probability $1 - k^{-\Omega(1)}$, and recovery takes $O(k)$ time.
This can be achieved for arbitrarily small $\eps > 0$, using
$O(k/\eps^2)$ rows.  We achieve a similar result in the $\ell_1$ norm.

The set query problem is useful in scenarios when, given a sketch of
$x$, we have some alternative methods for discovering a ``good''
support of an approximation to $x$. This is the case, e.g., in
block-sparse recovery, where (as we show in this paper) it is possible
to identify ``heavy'' blocks using other methods. It is also a natural
problem in itself. In particular, it generalizes the well-studied
\emph{point query problem}~\cite{CM03b}, which considers the case that
$S$ is a singleton. We note that, although the set query problem for
sets of size $k$ can be reduced to $k$ instances of the point query
problem, this reduction is less space-efficient than the algorithm we
propose, as elaborated below.

\textbf{Techniques.}
Our method is related to existing sparse recovery algorithms,
including Count-Sketch~\cite{ccf} and
Count-Min~\cite{CM03b}.  In fact, our sketch matrix $A$ is almost
identical to the one used in Count-Sketch---each column of $A$ has $d$
random locations out of $O(kd)$ each independently set to $\pm 1$, and
the columns are independently generated.  We can view such a matrix as
``hashing'' each coordinate to $d$ ``buckets'' out of $O(kd)$.  The
difference is that the previous algorithms require $O(k\log k)$
measurements to achieve our error bound (and $d = O(\log k)$), while
we only need $O(k)$ measurements and $d = O(1)$.

We overcome two obstacles to bring $d$ down to $O(1)$ and still
achieve the error bound with high probability\footnote{In this paper,
  ``high probability'' means probability at least $1 - 1/k^{c}$ for
  some constant $c > 0$.}.  First, in order to estimate the coordinates
$x_i$, we need a more elaborate method than, say, taking the median of
the buckets that $i$ was hashed into.  This is because, with constant
probability, all such buckets might contain some other elements from
$S$ (be ``heavy'') and therefore using {\em any} of them as an
estimator for $y_i$ would result in too much error.  Since, for
super-constant values of $|S|$, it is highly likely that such an event
will occur for at least one $i \in S$, it follows that this type of
estimation results in large error.

We solve this issue by using our knowledge of $S$.  We know when a
bucket is ``corrupted'' (that is, contains more than one element of
$S$), so we only estimate coordinates that lie in a large number of
uncorrupted buckets.  Once we estimate a coordinate, we subtract our
estimation of its value from the buckets it is contained in.  This
potentially decreases the number of corrupted buckets, allowing us to
estimate more coordinates.  We show that, with high probability, this
procedure can continue until it estimates every coordinate in $S$.

The other issue with the previous algorithms is that their analysis of
their probability of success does not depend on $k$.  This means that,
even if the ``head'' did not interfere, their chance of success would
be a constant (like $1 - 2^{-\Omega(d)}$) rather than high probability in
$k$ (meaning $1 - k^{-\Omega(d)}$).  We show that the errors in our estimates
of coordinates have low covariance, which allows us to apply
Chebyshev's inequality to get that the total error is concentrated
around the mean with high probability.

\textbf{Related work.}
A similar recovery algorithm (with $d=2$) has been analyzed and
applied in a streaming context in~\cite{EG}. However, in that paper
the authors only consider the case where the vector $y$ is
$k$-sparse. In that case, the termination property alone suffices,
since there is no error to bound. Furthermore, because $d=2$ they only
achieve a constant probability of success.  In this paper we consider
general vectors $y$ so we need to make sure the error remains bounded,
and we achieve a high probability of success.

The recovery procedure also has similarities to recovering LDPCs using
belief propagation, especially over the binary erasure channel.  The
similarities are strongest for exact recovery of $k$-sparse $y$; our
method for bounding the error from noise is quite different.

%XXX
%An important part of our method is that $A$ is sparse.  This allows us
%to compute $Ax$ and perform recovery in linear time.  If we are
%willing to relax this requirement, however, the problem becomes much
%simpler.  We can instead choose $A$ to be an i.i.d. random Gaussian
%matrix with $Ck$ rows for a constant $C$.  Then the 2-norm condition
%number of any $k$ columns of $A$ is constant with high
%probability~\cite{CD04}.  This means the $y$ that minimizes
%$\norm{2}{A_Sy - (Ax + \nu)}$ is a good approximation to $x_S$.
%However, both computing $Ax$ and finding $y$ are expensive when $k$ is
%large.

\textbf{Applications.} Our efficient solution to the set query problem
can be combined with existing techniques to achieve sparse recovery
under several models.

We say that a vector $x$ follows a \emph{Zipfian} or \emph{power law}
distribution with parameter $\alpha$ if $\abs{x_{r(i)}} =
\Theta(\abs{x_{r(1)}}i^{-\alpha})$ where $r(i)$ is the location of the
$i$th largest coefficient in $x$.  When $\alpha > 1/2$, $x$ is well
approximated in the $\ell_2$ norm by its sparse approximation.
Because a wide variety of real world signals follow power law
distributions (\cite{M04,broder}), this notion (related to
``compressibility''\footnote{A signal is ``compressible'' when
  $\abs{x_{r(i)}} = O(\abs{x_{r(1)}}i^{-\alpha})$ rather than
  $\Theta(\abs{x_{r(1)}}i^{-\alpha})$~\cite{CT06}. This allows it to
  decay very quickly then stop decaying for a while; we require that
  the decay be continuous.}) is often considered to be much of the
reason why sparse recovery is interesting \cite{CT06,cevher}.  Prior
to this work, sparse recovery of power law distributions has only been
solved via general sparse recovery methods: $(1+\eps)\err{2}{k}{x}$
error in $O(k\log (n/k))$ measurements.

However, locating the large coefficients in a power law distribution
has long been easier than in a general distribution.  Using $O(k \log
n)$ measurements, the Count-Sketch algorithm~\cite{ccf} can produce a
candidate set $S\subseteq \{1, \dotsc, b\}$ with $\abs{S} = O(k)$ that
includes all of the top $k$ positions in a power law distribution with
high probability (if $\alpha > 1/2$).  We can then apply our set query
algorithm to recover an approximation $x'$ to $x_S$.  Because we
already are using $O(k \log n)$ measurements on Count-Sketch, we use
$O(k \log n)$ rather than $O(k)$ measurements in the set query
algorithm to get an $\eps / \sqrt{\log n}$ rather than $\eps$
approximation.  This lets us recover a $k$-sparse $x'$ with $O(k \log
n)$ measurements with
\[
\norm{2}{x' - x} \leq \left(1 + \frac{\eps}{\sqrt{\log n}}\right) \err{2}{k}{x}.
\]
This is especially interesting in the common regime where $k < n^{1 -
  c}$ for some constant $c > 0$.  Then no previous algorithms achieve
better than a $(1 + \eps)$ approximation with $O(k \log n)$
measurements, and the lower bound in~\cite{dipw} shows that any $O(1)$
approximation requires $\Omega(k \log n)$ measurements\footnote{The
  lower bound only applies to geometric distributions, not Zipfian
  ones.  However, our algorithm applies to more general
  \emph{sub-Zipfian} distributions (defined in
  Section~\ref{sec:zipfian}), which includes both.}.  This means at
$\Theta(k \log n)$ measurements, the best approximation changes from
$\omega(1)$ to $1 + o(1)$.

Another application is that of finding {\em block-sparse}
approximations.  In this application, the coordinate set $\{1 \ldots
n\}$ is partitioned into $n/b$ blocks, each of length $b$.  We define
a $(k, b)$-block-sparse vector to be a vector where all non-zero
elements are contained in at most $k/b$ blocks.  An example of
block-sparse data is time series data from $n/b$ locations over $b$
time steps, where only $k/b$ locations are ``active''.  We can define
\[ \err{2}{k,b}{x} = \min_{(k, b)-\smbox{block-sparse }\hat{x}}\norm{2}{x-\hat{x}}.
\]

The block-sparse recovery problem can now be formulated analogously to
Equation~\ref{e:lplq}.  Since the formulation imposes restrictions on
the sparsity patterns, it is natural to expect that one can perform
sparse recovery from fewer than $O(k \log (n/k))$ measurements needed
in the general case.  Because of that reason and the prevalence of
approximately block-sparse signals, the problem of stable recovery of
variants of block-sparse approximations has been recently a subject of
extensive research~(e.g., see~\cite{EB,SPH,BCDH,CIHB}). The state of
the art algorithm has been given in~\cite{BCDH}, who gave a
probabilistic construction of a single $m \times n$ matrix $A$, with
$m=O(k +\frac{k}{b} \log n$), and an $n \log^{O(1)} n$-time algorithm
for performing the block-sparse recovery in the $\ell_1$ norm (as well
as other variants).  If the blocks have size $\Omega(\log n)$, the
algorithm uses only $O(k)$ measurements, which is a substantial
improvement over the general bound.  However, the approximation factor
$C$ guaranteed by that algorithm was super-constant.

In this paper, we provide a distribution over matrices $A$, with
$m=O(k + \frac{k}{b} \log n)$, which enables solving this problem with
a {\em constant} approximation factor and in the $\ell_2$ norm, with
high probability.  As with Zipfian distributions, first one algorithm
tells us where to find the heavy hitters and then the set query
algorithm estimates their values.  In this case, we modify the
algorithm of~\cite{block-heavy-hitters} to find {\em block heavy
  hitters}, which enables us to find the support of the $\frac{k}{b}$
``most significant blocks'' using $O(\frac{k}{b} \log n)$
measurements.  The essence is to perform dimensionality reduction of
each block from $b$ to $O(\log n)$ dimensions, then estimate the
result with a linear hash table.  For each block, most of the
projections are estimated pretty well, so the median is a good
estimator of the block's norm.  Once the support is identified, we can
recover the coefficients using the set query algorithm.

\section{Preliminaries}

\subsection{Notation}

For $n \in \Z^+$, we denote $\{1, \dotsc, n\}$ by $[n]$.  Suppose $x
\in \R^n$.  Then for $i \in [n]$, $x_i \in \R$ denotes the value of
the $i$-th coordinate in $x$.  As an exception, $e_i \in \R^n$ denotes
the elementary unit vector with a one at position $i$.  For $S
\subseteq [n]$, $x_S$ denotes the vector $x' \in R^n$ given by $x'_i =
x_i$ if $i \in S$, and $x'_i = 0$ otherwise.  We use $\supp(x)$ to
denote the support of $x$.  We use upper case letters to denote sets,
matrices, and random distributions.  We use lower case letters for
scalars and vectors.

\subsection{Negative Association}

This paper would like to make a claim of the form ``We have $k$
observations each of whose error has small expectation and variance.
Therefore the average error is small with high probability in $k$.''
If the errors were independent this would be immediate from
Chebyshev's inequality, but our errors depend on each other.
Fortunately, our errors have some tendency to behave even better than
if they were independent: the more noise that appears in one
coordinate, the less remains to land in other coordinates.  We use
\emph{negative dependence} to refer to this general class of behavior.
The specific forms of negative dependence we use are \emph{negative
  association} and \emph{approximate negative correlation}; see
Appendix~\ref{app:negative} for details on these notions.

\section{Set-Query Algorithm}

\begin{theorem}\label{main-theorem}
  There is a randomized sparse binary sketch matrix $A$ and recovery
  algorithm $\mathscr{A}$, such that for any $x \in \R^n$, $S
  \subseteq [n]$ with $\abs{S} = k$, $x' = \mathscr{A}(Ax + \nu, S)
  \in \R^n$ has $\supp(x') \subseteq S$ and
  \[
  \norm{2}{x' - x_S} \leq \eps(\norm{2}{x - x_S} + \norm{2}{\nu})
  \]
  with probability at least $1 - 1/k^c$.  $A$ has
  $O(\frac{c}{\eps^2}k)$ rows and $O(c)$ non-zero entries per column,
  and $\mathscr{A}$ runs in $O(ck)$ time.

  One can achieve $\norm{1}{x' - x_S} \leq \eps(\norm{1}{x - x_S} +
  \norm{1}{\nu})$ under the same conditions, but with only
  $O(\frac{c}{\eps}k)$ rows.
\end{theorem}

We will first show Theorem~\ref{main-theorem} for a constant $c = 1/3$
rather than for general $c$.  Parallel repetition gives the theorem
for general $c$, as described in Section~\ref{sec:wrapping}.  We
will also only show it with entries of $A$ being in $\{0, 1, -1\}$.
By splitting each row in two, one for the positive and one for the
negative entries, we get a binary matrix with the same properties.
The paper focuses on the more difficult $\ell_2$ result; see
Appendix~\ref{app:l1} for details on the $\ell_1$ result.

\subsection{Intuition}\label{sec:intuition}

We call $x_S$ the ``head'' and $x - x_S$ the ``tail.''  The head
probably contains the heavy hitters, with much more mass than the tail
of the distribution.  We would like to estimate $x_S$ with zero error
from the head and small error from the tail with high probability.

Our algorithm is related to the standard Count-Sketch~\cite{ccf} and
Count-Min~\cite{CM03b} algorithms.  In order to point out the
differences, let us examine how they perform on this task.  These
algorithms show that hashing into a single $w=O(k)$ sized hash table
is good in the sense that each point $x_i$ has:
\begin{enumerate}
\item Zero error from the head with constant probability (namely $1-\frac{k}{w}$).
\item A small amount of error from the tail in expectation (and hence
  with constant probability).
\end{enumerate}
They then iterate this procedure $d$ times and take the median, so
that each estimate has small error with probability $1 -
2^{-\Omega(d)}$.  With $d = O(\log k)$, we get that all $k$ estimates
in $S$ are good with $O(k \log k)$ measurements with high probability
in $k$.  With fewer measurements, however, some $x_i$ will probably
have error from the head.  If the head is much larger than the tail
(such as when the tail is zero), this is a major problem.
Furthermore, with $O(k)$ measurements the error from the tail would be
small only in expectation, not with high probability.

We make three observations that allow us to use only $O(k)$
measurements to estimate $x_S$ with error relative to the tail with
high probability in $k$.

\begin{enumerate}
\item The total error from the tail over a support of size $k$ is
  concentrated more strongly than the error at a single point:
  the error probability drops as $k^{-\Omega(d)}$ rather than $2^{-\Omega(d)}$.
\item The error from the head can be avoided if one knows where the head
  is, by modifying the recovery algorithm.
\item The error from the tail remains concentrated after modifying the
  recovery algorithm.
\end{enumerate}

For simplicity this paper does not directly show (1), only (2) and
(3).  The modification to the algorithm to achieve (2) is quite
natural, and described in detail and illustrated in
Section~\ref{sec:algorithm}.  Rather than estimate every coordinate in
$S$ immediately, we only estimate those coordinates which mostly do
not overlap with other coordinates in $S$.  In particular, we only
estimate $x_i$ as the median of at least $d-2$ positions that are not
in the image of $S \setminus \{i\}$.  Once we learn $x_i$, we can
subtract $Ax_ie_i$ from the observed $Ax$ and repeat on $A(x -
x_ie_i)$ and $S \setminus \{i\}$.  Because we only look at positions
that are in the image of only one remaining element of $S$, this
avoids any error from the head.  We show in Section~\ref{sec:exact}
that this algorithm never gets stuck; we can always find some position
that mostly doesn't overlap with the image of the rest of the
remaining support.

We then show that the error from the tail has low expectation, and
that it is strongly concentrated.  We think of the tail as noise
located in each ``cell'' (coordinate in the image space).  We
decompose the error of our result into two parts: the ``point error''
and the ``propagation''.  The point error is error introduced in our
estimate of some $x_i$ based on noise in the cells that we estimate
$x_i$ from, and equals the median of the noise in those cells.  The
``propagation'' is the error that comes from point error in estimating
other coordinates in the same connected component; these errors
propagate through the component as we subtract off incorrect estimates
of each $x_i$.

Section~\ref{sec:totalerror} shows how to decompose the total error in
terms of point errors and the component sizes.  The two following
sections bound the expectation and variance of these two quantities
and show that they obey some notions of negative dependence.  We
combine these errors in Section~\ref{sec:wrapping} to get
Theorem~\ref{main-theorem} with a specific $c$ (namely $c = 1/3$).  We
then use parallel repetition to achieve Theorem~\ref{main-theorem} for
arbitrary $c$.

\subsection{Algorithm}\label{sec:algorithm}

We describe the sketch matrix $A$ and recovery procedure in
Algorithm~\ref{algsupport}.  Unlike Count-Sketch~\cite{ccf} or
Count-Min~\cite{CM03b}, our $A$ is not split into $d$ hash tables of
size $O(k)$.  Instead, it has a single $w = O(d^2k/\eps^2)$ sized hash
table where each coordinate is hashed into $d$ unique positions.  We
can think of $A$ as a random $d$-uniform hypergraph, where the
non-zero entries in each column correspond to the terminals of a
hyperedge.  We say that $A$ is drawn from $\G^d(w, n)$ with random
signs associated with each (hyperedge, terminal) pair.  We do this so
we will be able to apply existing theorems on random hypergraphs.

Figure~\ref{fig:problem} shows an example $Ax$ for a given $x$, and
Figure~\ref{fig:algorithm} demonstrates running the recovery procedure
on this instance.

  \begin{figure}[]
    \centering
    \includegraphics[height=2in]{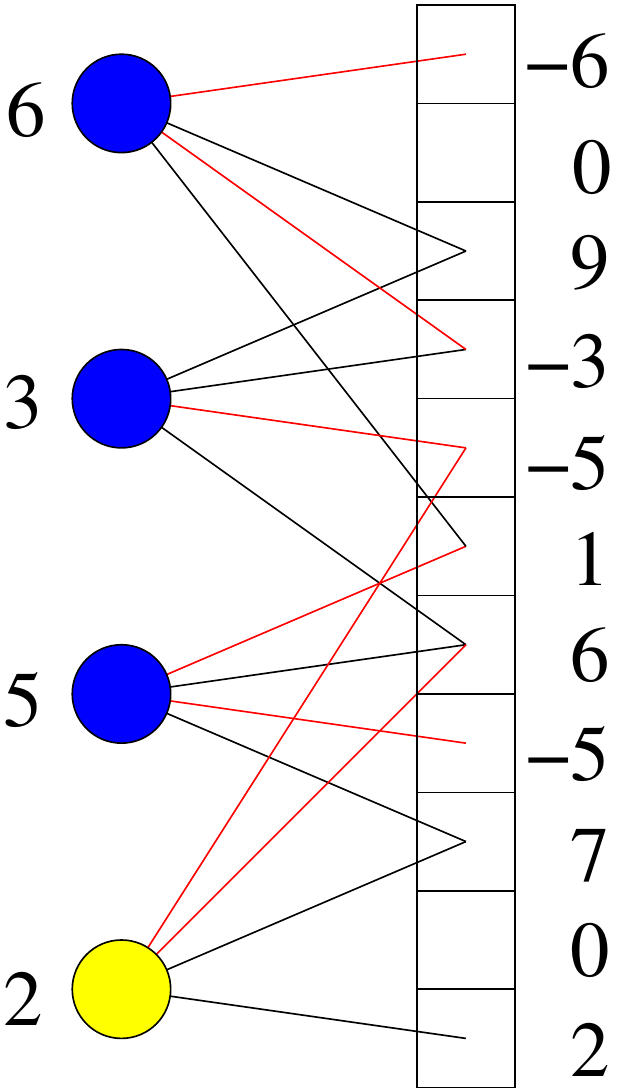}
    \caption{An instance of the set query problem.  There are $n$
      vertices on the left, corresponding to $x$, and the table on the
      right represents $Ax$.  Each vertex $i$ on the left maps to $d$
      cells on the right, randomly increasing or decreasing the value
      in each cell by $x_i$.  We represent addition by black lines,
      and subtraction by red lines.  We are told the locations of the
      heavy hitters, which we represent by blue circles; the rest is
      represented by yellow circles.}
    \label{fig:problem}
  \end{figure}
  \begin{figure}[]
    \centering
    \subfloat[][]{\includegraphics[height=2in]{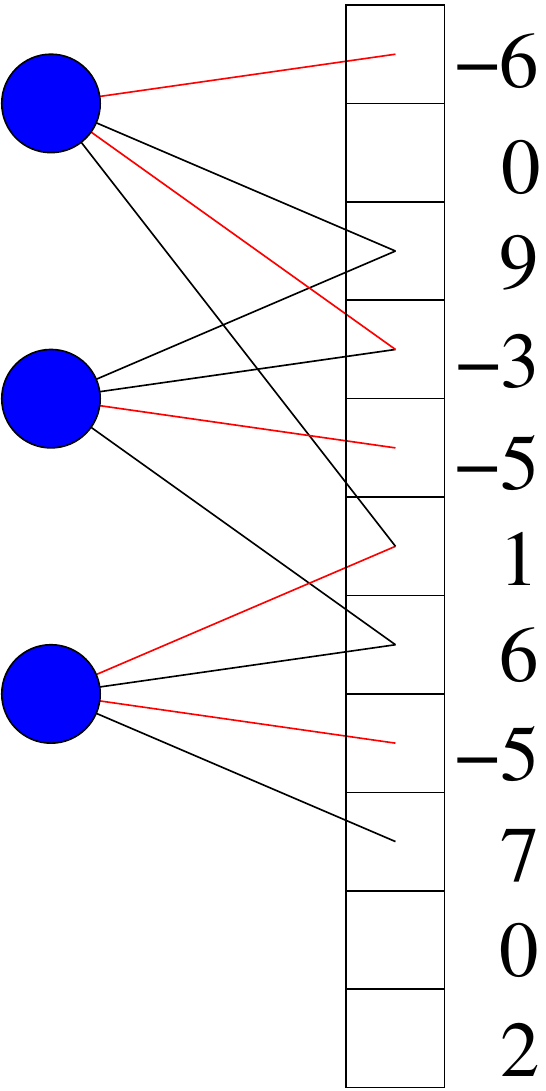}}
    \hspace{24pt}
    \subfloat[][]{\includegraphics[height=2in]{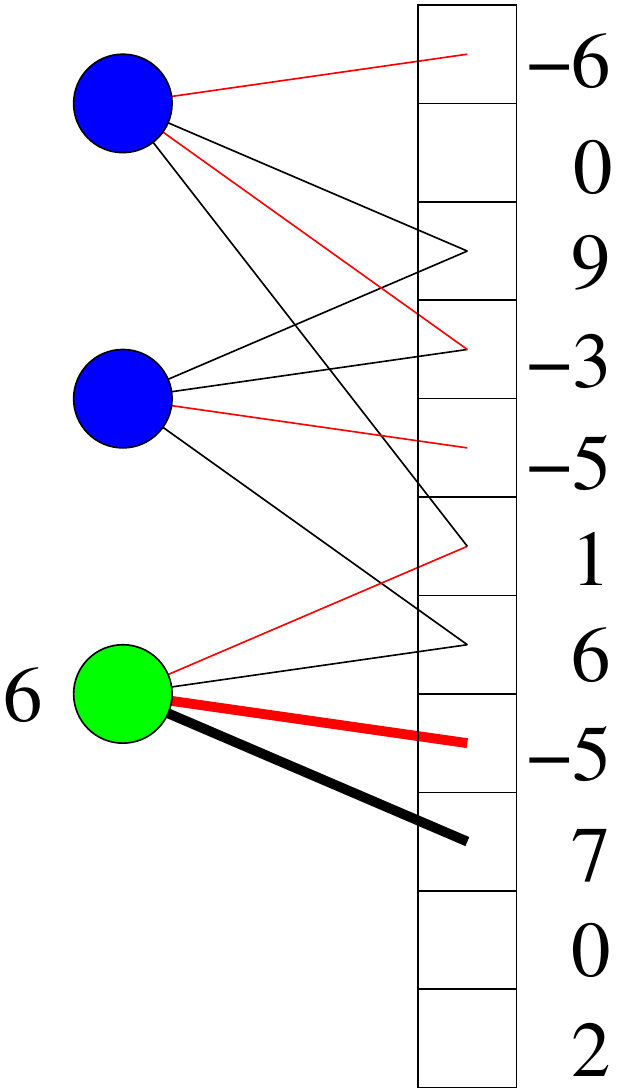}}
    \hspace{24pt}
    \subfloat[][]{\includegraphics[height=2in]{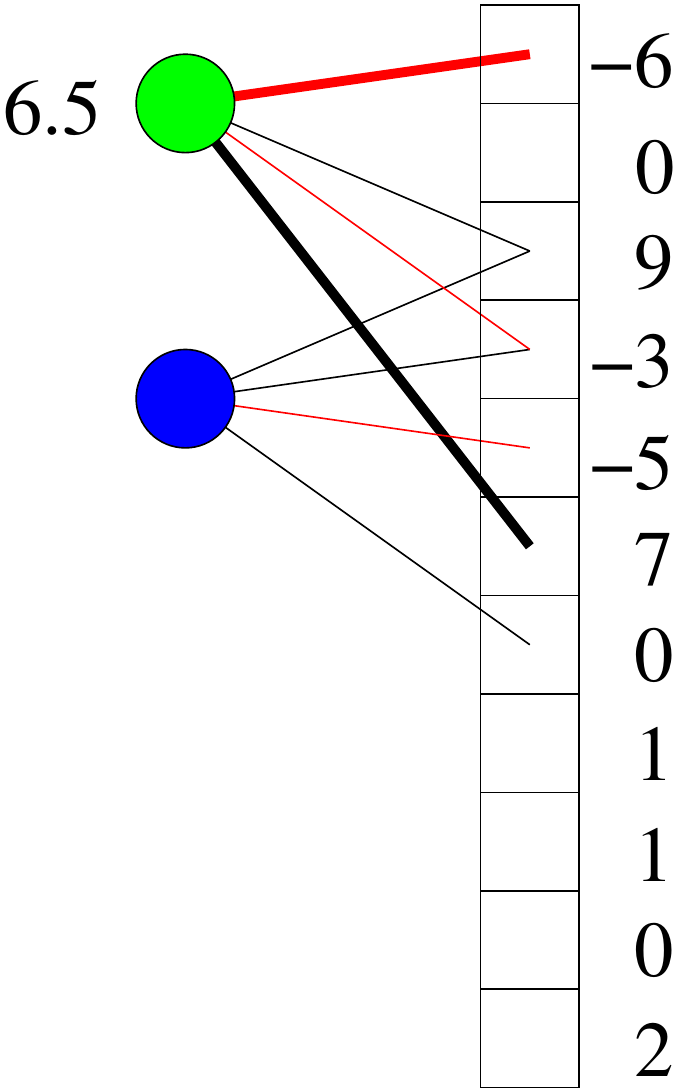}}
    \hspace{24pt}
    \subfloat[][]{\includegraphics[height=2in]{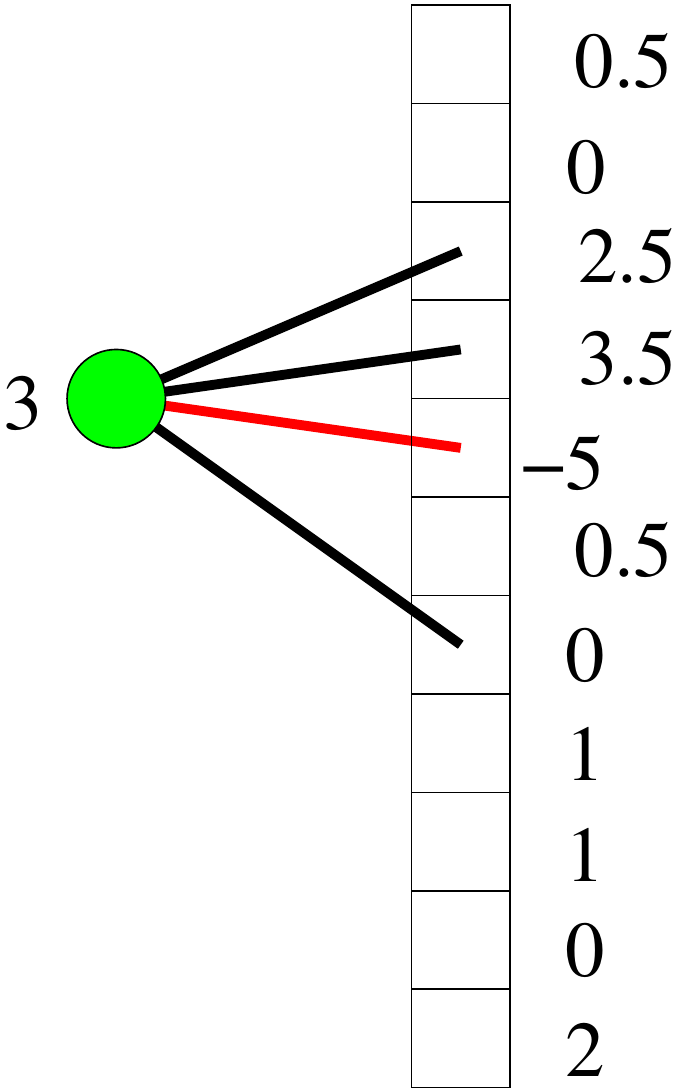}}
    \caption{Example run of the algorithm.  Part (a) shows the state
      as considered by the algorithm: $Ax$ and the graph structure
      corresponding to the given support.  In part (b), the algorithm
      chooses a hyperedge with at least $d-2$ isolated vertices and
      estimates the value as the median of those isolated vertices
      multiplied by the sign of the corresponding edge.  In part (c),
      the image of the first vertex has been removed from $Ax$ and we
      repeat on the smaller graph.  We continue until the entire
      support has been estimated, as in part (d).}
    \label{fig:algorithm}
  \end{figure}

  \begin{algorithm}
    \caption{Recovering a signal given its support.}\label{algsupport}
    \textbf{Definition of sketch matrix $A$.}  For a constant $d$, let
    $A$ be a $w \times n = O(\frac{d^2}{\eps^2} k) \times n$ matrix where each column
    is chosen independently uniformly at random over all exactly
    $d$-sparse columns with entries in $\{-1, 0, 1\}$.  We can think
    of $A$ as the incidence matrix of a random $d$-uniform hypergraph
    with random signs.

  \textbf{Recovery procedure.}
    \begin{algorithmic}[1]
      \Procedure{SetQuery}{$A, S, b$}\Comment{Recover
        approximation $x'$ to $x_S$ from $b=Ax+\nu$}
      \State $T \gets S$
      \While{$\abs{T} > 0$}
      \State Define $P(q) = \{j \mid A_{qj} \neq 0, j \in T\}$ as the set of hyperedges in $T$ that contain $q$.
      \State Define $L_j = \{q \mid A_{qj} \neq 0, \abs{P(q)} = 1\}$ as the set of isolated vertices in hyperedge $j$.
      \State Choose a random $j \in T$ such that $\abs{L_j} \geq d-1$.  If
      this is not possible, find a random $j \in T$ such that $\abs{L_j}
      \geq d-2$.  If neither is possible, abort.
      \State $x'_j \gets \median_{q \in L_j}A_{qj}b_q$
      \State $b \gets b - x'_jAe_j$
      \State $T \gets T \setminus \{j\}$
      \EndWhile
      \State \textbf{return} $x'$
      \EndProcedure
    \end{algorithmic}
  \end{algorithm}

  \begin{lemma}
    Algorithm~\ref{algsupport} runs in time $O(dk)$.
  \end{lemma}
  \begin{proof}
    $A$ has $d$ entries per column.  For each of the at most $dk$ rows
    $q$ in the image of $S$, we can store the preimages $P(q)$.  We
    also keep track of the sets of possible next hyperedges, $J_{i} =
    \{j \mid \abs{L_j} \geq d - i\}$ for $i \in \{1,2\}$.  We can
    compute these in an initial pass in $O(dk)$.  Then in each
    iteration, we remove an element $j \in J_1$ or $J_2$ and update
    $x'_j$, $b$, and $T$ in $O(d)$ time.  We then look at the two or
    fewer non-isolated vertices $q$ in hyperedge $j$, and remove $j$
    from the associated $P(q)$.  If this makes $\abs{P(q)} = 1$,
    we check whether to insert the element in $P(q)$ into the $J_i$.
    Hence the inner loop takes $O(d)$ time, for $O(dk)$ total.
  \end{proof}

\subsection{Exact Recovery}\label{sec:exact}

The random hypergraph $\G^d(w, k)$ of $k$ random $d$-uniform
hyperedges on $w$ vertices is well studied in \cite{randomhypergraph}.
We use their results to show that the algorithm successfully
terminates with high probability, and that most hyperedges are chosen
with at least $d-1$ isolated vertices:

\begin{lemma}\label{thm:termination}
  With probability at least $1 - O(1/k)$, Algorithm~\ref{algsupport}
  terminates without aborting.  Furthermore, in each component at most
  one hyperedge is chosen with only $d-2$ isolated vertices.
\end{lemma}

We will show this by building up a couple lemmas.  We define a
connected hypergraph $H$ with $r$ vertices on $s$ hyperedges to be a
\emph{hypertree} if $r = s(d-1)+1$ and to be \emph{unicyclic} if $r =
s(d-1)$.  Then Theorem~4 of \cite{randomhypergraph} shows that, if the
graph is sufficiently sparse, $\G^d(w, k)$ is probably composed
entirely of hypertrees and unicyclic components.  The precise
statement is as follows\footnote{Their statement of the theorem is
  slightly different.  This is the last equation in their proof of the
  theorem.}:

  \begin{lemma}[Theorem 4 of \cite{randomhypergraph}]
    Let $m = w/d(d-1) - k$.  Then with probability $1 - O(d^5 w^2/m^3)$,
    $\G^d(w, k)$ is composed entirely of hypertrees and unicyclic
    components.
  \end{lemma}
  We use a simple consequence:
  \begin{corollary}\label{cor:randomhyper}
    If $d = O(1)$ and $w \geq 2d(d-1)k$, then with probability $1 -
    O(1/k)$, $\G^d(w, k)$ is composed entirely of hypertrees and unicyclic
  \end{corollary}

  We now prove some basic facts about hypertrees and unicyclic
  components:

  \begin{lemma}\label{hypertreeunicycle}
    Every hypertree has a hyperedge incident on at least $d-1$ isolated
    vertices.  Every unicyclic component either has a hyperedge incident
    on $d-1$ isolated vertices or has a hyperedge incident on $d-2$
    isolated vertices, the removal of which turns the unicyclic
    component into a hypertree.
  \end{lemma}
  \begin{proof}
    Let $H$ be a connected component of $s$ hyperedges and $r$ vertices.

    If $H$ is a hypertree, $r = (d-1)s + 1$.  Because $H$ has only $ds$
    total (hyperedge, incident vertex) pairs, at most $2(s-1)$ of these
    pairs can involve vertices that appear in two or more hyperedges.
    Thus at least one of the $s$ edges is incident on at most one vertex
    that is not isolated, so some edge has $d-1$ isolated vertices.

    If $H$ is unicyclic, $r = (d-1)s$ and so at most $2s$ of the
    (hyperedge, incident vertex) pairs involve non-isolated vertices.
    Therefore on average, each edge has $d-2$ isolated vertices.  If no
    edge is incident on at least $d-1$ isolated vertices, every edge
    must be incident on exactly $d-2$ isolated vertices.  In that case,
    each edge is incident on exactly two non-isolated vertices and each
    non-isolated vertex is in exactly two edges.  Hence we can perform
    an Eulerian tour of all the edges, so removing any edge does not
    disconnect the graph.  After removing the edge, the graph has $s' =
    s-1$ edges and $r' = r - d + 2$ vertices; therefore $r' = (d-1)s' +
    1$ so the graph is a hypertree.
  \end{proof}

  Corollary~\ref{cor:randomhyper} and Lemma~\ref{hypertreeunicycle}
  combine to show Lemma~\ref{thm:termination}.
\subsection{Total error in terms of point error and component size}\label{sec:totalerror}

Define $C_{i,j}$ to be the event that hyperedges $i$ and $j$ are in
the same component, and $D_i = \sum_j C_{i,j}$ to be the number of
hyperedges in the same component as $i$.  Define $L_i$ to be the cells
that are used to estimate $i$; so $L_i = \{q \mid A_{qj} \neq 0,
\abs{P(q)} = 1\}$ at the round of the algorithm when $i$ is estimated.
Define $Y_i = \median_{q \in L_i} A_{qi}(b-Ax_S)_q$ to be the ``point
error'' for hyperedge $i$, and $x'$ to be the output of the algorithm.
Then the deviation of the output at any coordinate $i$ is at most
twice the sum of the point errors in the component containing $i$:

\begin{lemma}\label{thm:pointdeviation}
\[
  \abs{(x'-x_S)_i} \leq 2\sum_{j\in S} \abs{Y_j} C_{i,j}.
\]
\end{lemma}
\begin{proof}
  Let $T_i = (x'-x_S)_i$, and define $R_i = \{j \mid j \neq i, \exists
  q \in L_i \mbox{ s.t. } A_{qj} \neq 0\}$ to be the set of hyperedges
  that overlap with the cells used to estimate $i$.  Then from the
  description of the algorithm, it follows that
  \begin{align*}
    T_i &= \median_{q\in L_i} A_{qi}((b-Ax_S)_q - \sum_{j} A_{qj}T_j)\\
    \abs{T_i} &\leq \abs{Y_i} + \sum_{j \in R_i} \abs{T_j}.
  \end{align*}
  We can think of the $R_i$ as a directed acyclic graph (DAG), where
  there is an edge from $j$ to $i$ if $j \in R_i$.  Then if $p(i, j)$
  is the number of paths from $i$ to $j$,
  \[
  \abs{T_i} \leq \sum_j p(j, i) \abs{Y_i}.
  \]

  Let $r(i) = \abs{\{j \mid i \in R_j\}}$ be the outdegree of the DAG.
  Because the $L_i$ are disjoint, $r(i) \leq d - \abs{L_i}$.  From
  Lemma~\ref{thm:termination}, $r(i) \leq 1$ for all but one hyperedge
  in the component, and $r(i) \leq 2$ for that one.  Hence $p(i, j)
  \leq 2$ for any $i$ and $j$, giving the result.
\end{proof}

We use the following corollary:
\begin{corollary}\label{thm:totalerror}
\[
  \norm{2}{x'-x_S}^2 \leq 4\sum_{i\in S} D_i^2Y_i^2
\]
\end{corollary}
\begin{proof}
  \begin{align*}
    \norm{2}{x'-x_S}^2 &= \sum_{i\in S} (x'-x_S)_i^2
    \leq 4\sum_{i\in S} (\sum_{\substack{j\in S\\C_{i,j}=1}} \abs{Y_j})^2\\
    &\leq 4\sum_{i\in S} D_i\sum_{\substack{j\in S\\C_{i,j}=1}} \abs{Y_j}^2
    = 4\sum_{i\in S} D_i^2Y_i^2
  \end{align*}
where the second inequality is the power means inequality.
\end{proof}

The $D_j$ and $Y_j$ are independent from each other, since one depends
only on $A$ over $S$ and one only on $A$ over $[n]\setminus S$.
Therefore we can analyze them separately; the next two sections show
bounds and negative dependence results for $Y_j$ and $D_j$,
respectively.

\subsection{Bound on point error}\label{sec:pointerror}

Recall from Section~\ref{sec:totalerror} that based entirely on the
set $S$ and the columns of $A$ corresponding to $S$, we can identify
the positions $L_i$ used to estimate $x_i$.  We then defined the
``point error''
\[
Y_i = \median_{q \in L_i} A_{qi}(b-Ax_S)_q = \median_{q \in L_i} A_{qi}(A(x-x_S) + \nu)_q
\]
and showed how to relate the total error to the point error.  Here we
would like to show that the $Y_i$ have bounded moments and are
negatively dependent.  Unfortunately, it turns out that the $Y_i$ are
not negatively associated so it is unclear how to show negative
dependence directly.  Instead, we will define some other variables
$Z_i$ that are always larger than the corresponding $Y_i$.  We will
then show that the $Z_i$ have bounded moments and negative
association.

We use the term ``NA'' throughout the proof to denote negative
association.  For the definition of negative association and relevant
properties, see Appendix~\ref{app:negative}.

\begin{lemma}\label{thm:pointerror}
  Suppose $d \geq 7$ and define $\mu =
  O(\frac{\eps^2}{k}(\norm{2}{x-x_S}^2 + \norm{2}{\nu}^2))$.  There
  exist random variables $Z_i$ such that the variables $Y_i^2$ are
  stochastically dominated by $Z_i$, the $Z_i$ are negatively
  associated, $\E[Z_i] = \mu$, and $\E[Z_i^2] = O(\mu^2)$.
\end{lemma}

\begin{proof}
  The choice of the $L_i$ depends only on the values of $A$ over $S$;
  hence conditioned on knowing $L_i$ we still have $A(x-x_S)$
  distributed randomly over the space.  Furthermore the distribution
  of $A$ and the reconstruction algorithm are invariant under
  permutation, so we can pretend that $\nu$ is permuted randomly
  before being added to $Ax$.  Define $B_{i,q}$ to be the event that
  $q \in \supp(Ae_i)$, and define $D_{i,q} \in \{-1, 1\}$
  independently at random.  Then define the random variable
  \[
  V_q = (b-Ax_S)_q = \nu_q + \sum_{i \in [n]\setminus S} x_iB_{i,q}D_{i,q}.
  \]
  Because we want to show concentration of measure, we would like to
  show negative association (NA) of the $Y_i = \median_{q \in L_i}
  A_{qi}V_q$.  We know $\nu$ is a permutation distribution, so it is
  NA~\cite{joagdevproschan}.  The $B_{i,q}$ for each $i$ as a function
  of $q$ are chosen from a Fermi-Dirac model, so they are
  NA~\cite{Dubhashi96ballsand}.  The $B_{i,q}$ for different $i$ are
  independent, so all the $B_{i,q}$ variables are NA.  Unfortunately,
  the $D_{i,q}$ can be negative, which means the $V_q$ are not
  necessarily NA.  Instead we will find some NA variables that
  dominate the $V_q$.  We do this by considering $V_q$ as a
  distribution over $D$.

  Let $W_q = \E_{D} [V_q^2] = \nu_q^2 + \sum_{i \in [n] \setminus S}
  x_i^2B_{i,q}$.  As increasing functions of NA variables, the $W_q$
  are NA.  By Markov's inequality $\Pr_D[V_q^2 \geq c W_q] \leq
  \frac{1}{c}$, so after choosing the $B_{i,q}$ and as a distribution
  over $D$, $V_q^2$ is dominated by the random variable $U_q = W_qF_q$
  where $F_q$ is, independently for each $q$, given by the
  p.d.f. $f(c) = 1/c^2$ for $c \geq 1$ and $f(c) = 0$ otherwise.
  Because the distribution of $V_q$ over $D$ is independent for each
  $q$, the $U_q$ jointly dominate the $V_q^2$.

  The $U_q$ are the componentwise product of the $W_q$ with independent
  positive random variables, so they too are NA.
  Then define
  \[
  Z_i = \median_{q \in L_i} U_q.
  \]
  As an increasing function of disjoint subsets of NA variables, the
  $Z_i$ are NA.  We also have that
  \begin{align*}
    Y_i^2 &= (\median_{q \in L_i} A_{qi}V_q)^2 \leq (\median_{q \in L_i} \abs{V_q})^2 \\
    &= \median_{q \in L_i} V_q^2 \leq \median_{q \in L_i} U_q = Z_i 
  \end{align*}
  so the $Z_i$ stochastically dominate $Y_i^2$.  We now will bound
  $\E[Z_i^2]$.  Define
  \begin{align*}
    \mu &= E[W_q] = \E[\nu_q^2] + \sum_{i \in [n] \setminus S} x_i^2E[B_{i,q}] \\
    &= \frac{d}{w}\norm{2}{x-x_S}^2 + \frac{1}{w}\norm{2}{\nu}^2 \\
    &\leq \frac{\eps^2}{k}(\norm{2}{x-x_S}^2 + \norm{2}{\nu}^2).
  \end{align*}
  Then we have
  \begin{align*}
    \Pr[W_q \geq c\mu] &\leq \frac{1}{c}\\
    \Pr[U_q \geq c \mu] &= \int_0^\infty  f(x)\Pr[W_q \geq c\mu/x]dx\\
    &\leq \int_1^c \frac{1}{x^2} \frac{x}{c}dx + \int_c^\infty \frac{1}{x^2}dx
    = \frac{1 + \ln c}{c}
  \end{align*}
  Because the $U_q$ are NA, they satisfy marginal probability
  bounds~\cite{Dubhashi96ballsand}:
  \[
  \Pr[U_q \geq t_q, q \in [w]] \leq \prod_{i \in [n]} \Pr[U_q \geq t_q]
  \]
  for any $t_q$.  Therefore
  \begin{align}
    \notag \Pr[Z_i \geq c\mu] &\leq \sum_{
      \substack{T \subset L_i\\\abs{T} = \abs{L_i} / 2}
    } \prod_{q \in T} Pr[U_q \geq c\mu]\\
    \notag &\leq 2^{\abs{L_i}}\left(\frac{1+\ln c}{c}\right)^{\abs{L_i}/2}\\
    \Pr[Z_i \geq c\mu] &\leq \left(4\frac{1+\ln c}{c}\right)^{d/2 - 1}
  \end{align}
  If $d \geq 7$, this makes $\E[Z_i] = O(\mu)$ and $\E[Z_i^2] =
  O(\mu^2)$.
\end{proof}
\subsection{Bound on component size}\label{sec:componentsize}

\begin{lemma}\label{thm:componentsize}
  Let $D_i$ be the number of hyperedges in the same component as hyperedge
  $i$.  Then for any $i \neq j$,
  \[
  \mbox{Cov}(D_i^2, D_j^2) = \E[D_i^2D_j^2] - \E[D_i^2]^2 \leq
  O(\frac{\log^6 k}{\sqrt{k}}).
  \]
  Furthermore, $\E[D_i^2] = O(1)$ and $\E[D_i^4] = O(1)$.
\end{lemma}

\begin{proof}
  The intuition is that if one component gets larger, other components
  tend to get smaller.  Also the graph is very sparse, so component
  size is geometrically distributed.  There is a small probability
  that $i$ and $j$ are connected, in which case $D_i$ and $D_j$ are
  positively correlated, but otherwise $D_i$ and $D_j$ should be
  negatively correlated.  However analyzing this directly is rather
  difficult, because as one component gets larger, the remaining
  components have a lower average size but higher variance.  Our
  analysis instead takes a detour through the hypergraph where each
  hyperedge is picked independently with a probability that gives the
  same expected number of hyperedges.  This distribution is easier to
  analyze, and only differs in a relatively small $\tilde O(\sqrt{k})$
  hyperedges from our actual distribution.  This allows us to move
  between the regimes with only a loss of $\tilde
  O(\frac{1}{\sqrt{k}})$, giving our result.

  Suppose instead of choosing our hypergraph from $\G^d(w, k)$ we
  chose it from $\G^d(w, \frac{k}{\binom{w}{d}})$; that is, each
  hyperedge appeared independently with the appropriate probability to
  get $k$ hyperedges in expectation.  This model is somewhat simpler,
  and yields a very similar hypergraph $\overline{G}$.  One can then modify
  $\overline{G}$ by adding or removing an appropriate number of random
  hyperedges $I$ to get exactly $k$ hyperedges, forming a uniform $G \in \G^d(w,
  k)$.  By the Chernoff bound, $\abs{I} \leq O(\sqrt{k}\log k)$ with
  probability $1 - \frac{1}{k^{\Omega(1)}}$.

  Let $\overline{D}_i$ be the size of the component containing $i$ in
  $\overline{G}$, and $H_i = D_i^2 - \overline{D}_i^2$.  Let $E$
  denote the event that any of the $D_i$ or $\overline{D}_i$ is more
  than $C\log k$, or that more than $C\sqrt{k} \log k$ hyperedges lie in
  $I$, for some constant $C$.  Then $E$ happens with probability less
  than $\frac{1}{k^5}$ for some $C$, so it has negligible influence on
  $\E[D_i^2D_j^2]$.  Hence the rest of this proof will assume $E$ does
  not happen.

  Therefore $H_i = 0$ if none of the $O(\sqrt{k}\log k)$ random
  hyperedges in $I$ touch the $O(\log k)$ hyperedges in the components
  containing $i$ in $\overline{G}$, so $H_i = 0$ with probability at
  least $1 - O(\frac{\log^2 k}{\sqrt{k}})$.  Even if $H_i \neq 0$, we
  still have $\abs{H_i} \leq (D_i^2 + D_j^2) \leq O(\log^2 k)$.

  Also, we show that the $\overline{D}_i^2$ are negatively correlated,
  when conditioned on being in separate components.  Let
  $\overline{D}(n, p)$ denote the distribution of the component size
  of a random hyperedge on $\G^d(n, p)$, where $p$ is the probability an
  hyperedge appears.  Then $\overline{D}(n, p)$ dominates $\overline{D}(n',
  p)$ whenever $n > n'$ --- the latter hypergraph is contained within the
  former.  If $\overline{C}_{i,j}$ is the event that $i$ and $j$ are
  connected in $\overline{G}$, this means
  \[
  \E[\overline{D}_i^2 \mid \overline{D}_j=t, \overline{C}_{i,j} = 0]
  \]
  is a decreasing function in $t$, so we have negative correlation:
  \begin{align*}
    \E[\overline{D}_i^2\overline{D}_j^2\mid \overline{C}_{i,j}=0] &\leq \E[\overline{D}_i^2\mid \overline{C}_{i,j}=0]\E[\overline{D}_j^2\mid \overline{C}_{i,j}=0] \\
    &\leq \E[\overline{D}_i^2]\E[\overline{D}_j^2].
  \end{align*}
  
  Furthermore for $i \neq j$, $\Pr[\overline{C}_{i,j}=1] = \E[\overline{C}_{i,j}] = \frac{1}{k-1}\sum_{l \neq i}\E[\overline{C}_{i,l}] = \frac{\E[\overline{D}_i] -
    1}{k-1} = O(1/k)$.  Hence
  \begin{align*}
    \E[\overline{D}_i^2\overline{D}_j^2] =& \E[\overline{D}_i^2\overline{D}_j^2\mid \overline{C}_{i,j}=0]\Pr[\overline{C}_{i,j}=0] + \\&\E[\overline{D}_i^2\overline{D}_j^2\mid \overline{C}_{i,j}=1]\Pr[\overline{C}_{i,j}=1]\\
    \leq& \E[\overline{D}_i^2]\E[\overline{D}_j^2] + O(\frac{\log^4 k}{k}).
  \end{align*}

  Therefore
  \begin{align*}
    &\E[D_i^2D_j^2] \\
    =& \E[(\overline{D}_i^2 + H_i)(\overline{D}_j^2 + H_j)]\\
    =& \E[\overline{D}_i^2\overline{D}_j^2] + 2\E[H_i\overline{D}_j^2] + \E[H_iH_j]\\
    \leq& \E[\overline{D}_i^2]\E[\overline{D}_j^2] + O(2\frac{\log^2 k}{\sqrt{k}} \log^4 k + \frac{\log^2 k}{\sqrt{k}}\log^2 k)\\
    =& \E[D_i^2 - H_i]^2 + O(\frac{\log^6 k}{\sqrt{k}})\\
    =& \E[D_i^2]^2 - 2\E[H_i]\E[D_i^2] + \E[H_i]^2 + O(\frac{\log^6 k}{\sqrt{k}})\\
    \leq& \E[D_i^2]^2 + O(\frac{\log^6 k}{\sqrt{k}})
  \end{align*}

  Now to bound $\E[D_i^4]$ in expectation.  Because our hypergraph is
  exceedingly sparse, the size of a component can be bounded by a
  branching process that dies out with constant probability at each
  step.  Using this method, Equations~71 and 72 of \cite{hyper10} state
  that $\Pr[\overline{D} \geq k] \leq e^{-\Omega(k)}$.  Hence
  $\E[\overline{D}_i^2] = O(1)$ and $\E[\overline{D}_i^4] = O(1)$.
  Because $H_i$ is $0$ with high probability and $O(\log^2 k)$
  otherwise, this immediately gives $\E[D_i^2] = O(1)$ and $\E[D_i^4]
  = O(1)$.
\end{proof}

\subsection{Wrapping it up}\label{sec:wrapping}

Recall from Corollary~\ref{thm:totalerror} that our total error
\[
\norm{2}{x'-x_S}^2 \leq 4\sum_i Y_i^2D_i^2 \leq 4\sum_i Z_iD_i^2.
\]

The previous sections show that $Z_i$ and $D_i^2$ each have small
expectation and covariance.  This allows us to apply Chebyshev's
inequality to concentrate $4\sum_i Z_iD_i^2$ about its expectation,
bounding $\norm{2}{x'-x_S}$ with high probability:

\begin{lemma}\label{partial-lemma}
  We can recover $x'$ from $Ax + \nu$ and $S$ with
  \[
  \norm{2}{x' - x_S} \leq \eps (\norm{2}{x-x_S} + \norm{2}{\nu})
  \]
  with probability at least $1 - \frac{1}{c^2k^{1/3}}$ in $O(k)$
  recovery time.  Our $A$ has $O(\frac{c}{\eps^2}k)$ rows and sparsity
  $O(1)$ per column.
\end{lemma}

\begin{proof}
  Our total error is
  \[
  \norm{2}{x'-x_S}^2 \leq 4\sum_i Y_i^2D_i^2 \leq 4\sum_i Z_iD_i^2.
  \]

  Then by Lemma~\ref{thm:pointerror} and Lemma~\ref{thm:componentsize},
  \begin{align*}
    \E[4\sum_i Z_iD_i^2] = 4 \sum_i \E[Z_i]\E[D_i^2] = k\mu
  \end{align*}
  where $\mu = O(\frac{\eps^2}{k}(\norm{2}{x - x_S}^2 + \norm{2}{\nu}^2))$.  Furthermore,
  \begin{align*}
    \E[(\sum_i Z_iD_i^2)^2] &= \sum_i \E[Z_i^2D_i^4] + \sum_{i \neq j} \E[Z_iZ_jD_i^2D_j^2]\\
    &= \sum_i \E[Z_i^2]\E[D_i^4] + \sum_{i \neq j} \E[Z_iZ_j]\E[D_i^2D_j^2]\\
    &\leq \sum_i O(\mu^2)  + \sum_{i \neq j} \E[Z_i]\E[Z_j](\E[D_i^2]^2 + O(\frac{\log^6 k}{\sqrt{k}}))\\
    &= O(\mu^2k\sqrt{k}\log^6k ) + k(k-1)\E[Z_iD_i^2]^2\\
    \mbox{Var}(\sum_i Z_iD_i^2) &= \E[(\sum_i Z_iD_i^2)^2] - k^2\E[Z_iD_i^2]^2\\
    &\leq O(\mu^2k\sqrt{k}\log^6k )
  \end{align*}
  By Chebyshev's inequality, this means
  \begin{align*}
    \Pr[4\sum_i Z_iD_i^2 \geq (1+c)\mu k] \leq O(\frac{\log^6 k}{c^2\sqrt{k}})\\
    \Pr[\norm{2}{x'-x_S}^2 \geq (1+c)C\eps^2(\norm{2}{x-x_S}^2 + \norm{2}{\nu}^2)] \leq O(\frac{1}{c^2k^{1/3}})
  \end{align*}
  for some constant $C$.
  Rescaling $\eps$ down by $\sqrt{C(1+c)}$, we can get
  \[
  \norm{2}{x' - x_S} \leq \eps (\norm{2}{x-x_S} + \norm{2}{\nu})
  \]
  with probability at least $1 - \frac{1}{c^2k^{1/3}}$:
\end{proof}

Now we shall go from $k^{-1/3}$ probability of error to $k^{-c}$ error
for arbitrary $c$, with $O(c)$ multiplicative cost in time and space.
We simply perform Lemma~\ref{partial-lemma} $O(c)$ times in parallel,
and output the pointwise median of the results.  By a standard
parallel repetition argument, this gives our main result:

\newtheorem*{oldtheorem}{Theorem \ref{main-theorem}}
\begin{oldtheorem}
  We can recover $x'$ from $Ax + \nu$ and $S$ with
  \[
  \norm{2}{x' - x_S} \leq \eps (\norm{2}{x-x_S} + \norm{2}{\nu})
  \]
  with probability at least $1 - \frac{1}{k^{c}}$ in $O(ck)$ recovery
  time.  Our $A$ has $O(\frac{c}{\eps^2}k)$ rows and sparsity $O(c)$
  per column.
\end{oldtheorem}
\begin{proof}
  Lemma~\ref{partial-lemma} gives an algorithm that achieves
  $O(k^{-1/3})$ probability of error.  We will show here how to
  achieve $k^{-c}$ probability of error with a linear cost in $c$, via
  a standard parallel repetition argument.

  Suppose our algorithm gives an $x'$ such that $\norm{2}{x'-x_S} \leq
  \mu$ with probability at least $1 - p$, and that we run this algorithm
  $m$ times independently in parallel to get output vectors $x^1, \dotsc,
  x^m$.  We output $y$ given by $y_i = \median_{j \in [m]} (x^j)_i$,
  and claim that with high probability $\norm{2}{y-x_S} \leq
  \mu\sqrt{3}$.

  Let $J = \{j \in [m] \mid \norm{2}{x^j-x_S} \leq \mu\}$.  Each $j \in
  [m]$ lies in $J$ with probability at least $1-p$, so the chance that
  $\abs{J} \leq 3m/4$ is less than $\binom{m}{m/4}p^{m/4} \leq
  (4ep)^{m/4}$.  Suppose that $\abs{J} \geq 3m/4$.  Then for all $i \in
  S$, $\abs{\{j \in J \mid (x^j)_i \leq y_i\}} \geq \abs{J} - \frac{m}{2} \geq \abs{J}/3$ and
  similarly $\abs{\{j \in J \mid (x^j)_i \geq y_i\}}\geq \abs{J}/3$.  Hence for
  all $i \in S$, $\abs{y_i - x_i}$ is smaller than at least $\abs{J}/3$
  of the $\abs{(x^j)_i - x_i}$ for $j \in J$.  Hence
  \begin{align*}
    \abs{J} \mu^2 &\geq
    \sum_{i \in S} \sum_{j \in J} ((x^j)_i - x_i)^2 \geq
    \sum_{i \in S} \frac{\abs{J}}{3}(y_i - x_i)^2 \\
    &= \frac{\abs{J}}{3} \norm{2}{y - x}^2
  \end{align*}
  or
  \[
  \norm{2}{y-x} \leq \sqrt{3}\mu
  \]
  with probability at least $1 - (4ep)^{m/4}$.

  Using Lemma~\ref{partial-lemma} to get $p = \frac{1}{16k^{1/3}}$ and
  $\mu = \eps(\norm{2}{x-x_S} + \norm{2}{\nu})$, with $m=12c$ repetitions
  we get Theorem~\ref{main-theorem}.
\end{proof}

\section{Applications}

We give two applications where the set query algorithm is a useful
primitive.

\subsection{Heavy Hitters of sub-Zipfian distributions}\label{sec:zipfian}

For a vector $x$, let $r_i$ be the index of the $i$th largest element,
so $\abs{x_{r_i}}$ is non-increasing in $i$.  We say that $x$ is
\emph{Zipfian with parameter $\alpha$} if $\abs{x_{r_i}} =
\Theta(\abs{x_{r_1}}i^{-\alpha})$.  We say that $x$ is
\emph{sub-Zipfian with parameters ($k$, $\alpha$)} if there exists a
non-increasing function $f$ with $\abs{x_{r_i}} =
\Theta(f(i)i^{-\alpha})$ for all $i \geq k$.  A Zipfian with parameter
$\alpha$ is a sub-Zipfian with parameter $(k, \alpha)$ for all $k$,
using $f(i) = \abs{x_{r_1}}$.

The Zipfian heavy hitters problem is, given a linear sketch $Ax$ of a
Zipfian $x$ and a parameter $k$, to find a $k$-sparse $x'$ with
minimal $\norm{2}{x-x'}$ (up to some approximation factor).  We
require that $x'$ be $k$-sparse (and no more) because we want to find
the heavy hitters themselves, not to find them as a proxy for
approximating $x$.

Zipfian distributions are common in real-world data sets, and finding
heavy hitters is one of the most important problems in data streams.
Therefore this is a very natural problem to try to improve; indeed,
the original paper on Count-Sketch discussed it~\cite{ccf}.  They show
a result complementary to our work, namely that one can find the
support efficiently:

\begin{lemma}[Section 4.1 of \cite{ccf}]\label{zipfiansupport}
  If $x$ is sub-Zipfian with parameter $(k, \alpha)$ and $\alpha >
  1/2$, one can recover a candidate support set $S$ with $\abs{S} =
  O(k)$ from $Ax$ such that $\{r_1, \dotsc, r_k\} \subseteq S$.  $A$
  has $O(k \log n)$ rows and recovery succeeds with high probability
  in $n$.
\end{lemma}
\begin{proof}[Proof sketch] Let $S_k = \{r_1, \dotsc, r_k\}$.  With
  $O(\frac{1}{\eps^2}k \log n)$ measurements, Count-Sketch identifies
  each $x_i$ to within $\frac{\eps}{k}\norm{2}{x - x_{S_k}}$ with
  high probability.  If $\alpha > 1/2$, this is less than $\abs{x_{r_k}}/3$
  for appropriate $\eps$.  But $\abs{x_{r_{9k}}} \leq \abs{x_{r_k}}/3$.  Hence
  only the largest $9k$ elements of $x$ could be estimated as larger
  than anything in $x_{S_k}$, so the locations of the largest $9k$
  estimated values must contain $S_k$.
\end{proof}

It is observed in~\cite{ccf} that a two-pass algorithm could identify
the heavy hitters exactly.  However, with a single pass, no better
method has been known for Zipfian distributions than for arbitrary
distributions; in fact, the lower bound~\cite{dipw} on linear sparse
recovery uses a geometric (and hence sub-Zipfian) distribution.

As discussed in~\cite{ccf}, using Count-Sketch\footnote{Another
  analysis (\cite{CM05}) uses Count-Min to achieve a better polynomial
  dependence on $\eps$, but at the cost of using the $\ell_1$ norm.
  Our result is an improvement over this as well.} with
$O(\frac{k}{\eps^2} \log n)$ rows gets a $k$-sparse $x'$ with
\[
  \norm{2}{x' - x} \leq (1 + \eps)\err{2}{k}{x} = \Theta(\frac{\abs{x_{r_1}}}{\sqrt{2\alpha-1}}k^{1/2-\alpha}).
\]
where, as in Section~\ref{sec:introduction},
\[
\err{2}{k}{x} = \min_{k\text{-sparse } \hat{x}}\norm{2}{\hat{x}-x}.
\]

  The set query
algorithm lets us improve from a $1+\eps$ approximation to a $1+o(1)$
approximation.  This is not useful for approximating $x$, since
increasing $k$ is much more effective than decreasing $\eps$.
Instead, it is useful for finding $k$ elements that are quite close to
being the actual $k$ heavy hitters of $x$.

Na\"ive application of the set query algorithm to the output set of
Lemma~\ref{zipfiansupport} would only get a close $O(k)$-sparse
vector, not a $k$-sparse vector.  To get a $k$-sparse vector, we must
show a lemma that generalizes one used in the proof of sparse recovery
of Count-Sketch (first in~\cite{CM06}, but our description is more
similar to~\cite{GI}).

\begin{lemma}\label{lemma:l2thresholding}
  Let $x, x' \in \R^n$.  Let $S$ and $S'$ be the locations of the
  largest $k$ elements (in magnitude) of $x$ and $x'$, respectively.
  Then if
  \[
  \tag{*} \norm{2}{(x'-x)_{S \cup S'}} \leq \eps \err{2}{k}{x},
  \]
  for $\eps \leq 1$, we have
  \[
  \norm{2}{x'_{S'} - x} \leq (1 + 3\eps)\err{2}{k}{x}.
  \]
\end{lemma}
Previous proofs have shown the following weaker form:
\begin{corollary}\label{cor:linfthresholding}
  If we change the condition (*) to $\norm{\infty}{x'-x} \leq
  \frac{\eps}{\sqrt{2k}} \err{2}{k}{x}$, the same result holds.
\end{corollary}
The corollary is immediate from Lemma~\ref{lemma:l2thresholding} and
$\norm{2}{(x'-x)_{S \cup S'}} \leq \sqrt{\abs{S \cup S'}}\norm{\infty}{(x'-x)_{S \cup S'}}$.

% \newproof{@proof1}{Proof of Lemma~\ref{lemma:l2thresholding}} 
% \newenvironment{proof1}{\begin{@proof1}}{\end{@proof1}} 

\begin{proof}[Proof of Lemma~\ref{lemma:l2thresholding}]
  We have
  \begin{align}
    \label{eq:l2topktotal} \norm{2}{x'_{S'}-x}^2 &= \norm{2}{(x'-x)_{S'}}^2 +
    \norm{2}{x_{S\setminus S'}}^2 + \norm{2}{x_{[n]\setminus (S \cup S')}}^2
  \end{align}
  The tricky bit is to bound the middle term $\norm{2}{x_{S\setminus
      S'}}^2$.  We will show that it is not much larger than
  $\norm{2}{x_{S'\setminus S}}^2$.

  Let $d = \abs{S \setminus S'}$, and let $a$ be the $d$-dimensional
  vector corresponding to the absolute values of the coefficients of
  $x$ over $S \setminus S'$.  That is, if $S \setminus S' = \{j_1,
  \dots, j_d\}$, then $a_i = \abs{x_{j_i}}$ for $i \in [d]$.  Let $a'$
  be analogous for $x'$ over $S \setminus S'$, and let $b$ and $b'$ be
  analogous for $x$ and $x'$ over $S'\setminus S$, respectively.

  Let $E = \err{2}{k}{x} = \norm{2}{x-x_S}$.  We have
  \begin{align*}
    \norm{2}{x_{S\setminus S'}}^2 - \norm{2}{x_{S'\setminus S}}^2 &= \norm{2}{a}^2 - \norm{2}{b}^2\\
    &= (a - b) \cdot (a + b)\\
    &\leq \norm{2}{a-b}\norm{2}{a+b}\\
    &\leq \norm{2}{a-b}(2\norm{2}{b} + \norm{2}{a-b})\\
    &\leq \norm{2}{a-b}(2E + \norm{2}{a-b})
  \end{align*}
  So we should bound $\norm{2}{a-b}$.  We know that
  $\abs{\abs{p}-\abs{q}} \leq \abs{p - q}$ for all $p$ and $q$, so
  \begin{align*}
    \norm{2}{a-a'}^2 + \norm{2}{b-b'}^2 &\leq \norm{2}{(x-x')_{S \setminus S'}}^2 + \norm{2}{(x-x')_{S' \setminus S}}^2 \\
    &\leq \norm{2}{(x-x')_{S \cup S'}}^2 \leq \eps^2 E^2.
  \end{align*}
  We also know that $a - b$ and $b' - a'$ both contain all nonnegative
  coefficients.  Hence
  \begin{align*}
    \norm{2}{a-b}^2 &\leq \norm{2}{a-b + b'-a'}^2\\
    &\leq \left(\norm{2}{a-a'} + \norm{2}{b'-b}\right)^2\\
    &\leq 2\norm{2}{a-a'}^2 + 2\norm{2}{b-b'}^2\\
    &\leq 2\eps^2E^2\\
    \norm{2}{a-b} &\leq \sqrt{2}\eps E.
  \end{align*}
  Therefore
  \begin{align*}
    \norm{2}{x_{S\setminus S'}}^2 - \norm{2}{x_{S'\setminus S}}^2 &\leq \sqrt{2}\eps E(2E + \sqrt{2} \eps E)\\
    &\leq (2\sqrt{2} + 2)\eps E^2\\
    &\leq 5\eps E^2.
  \end{align*}
  Plugging into Equation~\ref{eq:l2topktotal}, and using
  $\norm{2}{(x'-x)_{S'}}^2 \leq \eps^2 E^2$,
  \begin{align*}
    \norm{2}{x'_{S'}-x}^2 &\leq \eps^2 E^2 + 5\eps E^2  + \norm{2}{x_{S'\setminus
        S}}^2 + \norm{2}{x_{[n]\setminus (S \cup S')}}^2\\
    &\leq 6\eps E^2 + \norm{2}{x_{[n] \setminus S}}^2\\
    &= (1 + 6\eps)E^2\\
    \norm{2}{x'_{S'}-x} &\leq (1 + 3\eps)E.
  \end{align*}
\end{proof}

With this lemma in hand, on Zipfian distributions we can get a
$k$-sparse $x'$ with a $1 + o(1)$ approximation factor.

\begin{theorem}
  Suppose $x$ comes from a sub-Zipfian distribution with parameter
  $\alpha > 1/2$.  Then we can recover a $k$-sparse $x'$ from $Ax$ with
  \[
  \norm{2}{x' - x} \leq \frac{\eps}{\sqrt{\log n}} \err{2}{k}{x}.
  \]
  with $O(\frac{c}{\eps^2}k \log n)$ rows and $O(n\log n)$ recovery time, with
  probability at least $1 - \frac{1}{k^c}$.
\end{theorem}
\begin{proof}
  By Lemma~\ref{zipfiansupport} we can identify a set $S$ of size
  $O(k)$ that contains all the heavy hitters.  We then run the set
  query algorithm of Theorem~\ref{main-theorem} with
  $\frac{\eps}{3\sqrt{\log n}}$ substituted for $\eps$.  This gives an
  $\hat{x}$ with
  \begin{align*}
    \norm{2}{\hat{x}-x_S} &\leq \frac{\eps}{3\sqrt{\log n}} \err{2}{k}{x}.
  \end{align*}
  Let $x'$ contain the largest $k$ coefficients of $\hat{x}$.  By
  Lemma~\ref{lemma:l2thresholding} we have
  \begin{align*}
    \norm{2}{x' - x} \leq (1 + \frac{\eps}{\sqrt{\log n}}) \err{2}{k}{x}.
  \end{align*}
\end{proof}

\subsection{Block-sparse vectors}

In this section we consider the problem of finding {\em block-sparse}
approximations.  In this case, the coordinate set $\{1 \ldots n\}$ is
partitioned into $n/b$ blocks, each of length $b$.  We define a $(k,
b)$-block-sparse vector to be a vector where all non-zero elements are
contained in at most $k/b$ blocks.  That is, we partition $\{1,\dotsc,
n\}$ into $T_i = \{(i-1)b + 1, \dotsc, ib\}$.  A vector $x$ is
$(k,b)$-block-sparse if there exist $S_1, \dotsc, S_{k/b} \in \{T_1,
\dotsc, T_{n/b}\}$ with $\supp(x) \subseteq \bigcup S_i$.  Define
\[ \err{2}{k,b}{x} = \min_{(k,b)-\smbox{block-sparse
  }\hat{x}}\norm{2}{x-\hat{x}}.
\]

Finding the support of block-sparse vectors is closely related to
finding block heavy hitters, which is studied for the $\ell_1$ norm
in~\cite{block-heavy-hitters}.  The idea is to perform dimensionality
reduction of each block into $\log n$ dimensions, then perform sparse
recovery on the resulting $\frac{k\log n}{b}$-sparse vector.  The
differences from previous work are minor, so we relegate the details to
Appendix~\ref{app:strengthenheavyhitters}.

\begin{lemma}\label{thm:block-heavy-hitters}
  For any $b$ and $k$, there exists a family of matrices $A$ with
  $O(\frac{k}{\eps^5b}\log n)$ rows and column sparsity $O(\frac{1}{\eps^2}\log n)$
  such that we can recover a support $S$ from $Ax$ in $O(\frac{n}{\eps^2
    b}\log n)$ time with
  \[
  \norm{2}{x-x_S} \leq (1+\eps)\err{2}{k,b}{x}
  \]
  with probability at least $1 - n^{-\Omega(1)}$.
\end{lemma}

Once we know a good support $S$, we can run Algorithm~\ref{algsupport}
to estimate $x_S$:

\begin{theorem}
  For any $b$ and $k$, there exists a family of binary matrices $A$
  with $O(\frac{1}{\eps^2}k + \frac{k}{\eps^5b}\log n)$ rows such that
  we can recover a $(k,b)$-block-sparse $x'$ in $O(k + \frac{n}{\eps^2
    b}\log n)$ time with
  \[
  \norm{2}{x'-x} \leq (1+\eps)\err{2}{k,b}{x}
  \]
  with probability at least $1 - \frac{1}{k^{\Omega(1)}}$.
\end{theorem}
\begin{proof}
  Let $S$ be the result of Lemma~\ref{thm:block-heavy-hitters} with
  approximation $\eps/3$, so
  \[
  \norm{2}{x-x_S} \leq (1+\frac{\eps}{3})\err{2}{k,b}{x}.
  \]
  Then the set query algorithm on $x$ and $S$ uses $O(k/\eps^2)$ rows
  to return an $x'$ with
  \[
  \norm{2}{x' - x_S} \leq \frac{\eps}{3}\norm{2}{x-x_S}.
  \]
  Therefore
  \begin{align*}
    \norm{2}{x'-x} &\leq \norm{2}{x' - x_S}  + \norm{2}{x - x_S}\\
    &\leq (1 + \frac{\eps}{3}) \norm{2}{x-x_S}\\
    &\leq (1 + \frac{\eps}{3})^2\err{2}{k,b}{x}\\
    &\leq (1 + \eps)\err{2}{k,b}{x}
  \end{align*}
  as desired.
\end{proof}

If the block size $b$ is at least $\log n$ and $\eps$ is constant,
this gives an optimal bound of $O(k)$ rows.

\section{Conclusion and Future Work}

We show efficient recovery of vectors conforming to Zipfian or block
sparse models, but leave open extending this to other models.  Our
framework decomposes the task into first locating the heavy hitters
and then estimating them, and our set query algorithm is an efficient
general solution for estimating the heavy hitters once found.  The
remaining task is to efficiently locate heavy hitters in other models.

Our analysis assumes that the columns of $A$ are fully independent.
It would be valuable to reduce the independence needed, and hence the
space required to store $A$.

We show $k$-sparse recovery of Zipfian distributions with $1 + o(1)$
approximation in $O(k \log n)$ space.  Can the $o(1)$ be made smaller,
or a lower bound shown, for this problem?

\section*{Acknowledgments}
I would like to thank my advisor Piotr Indyk for much helpful advice, Anna
Gilbert for some preliminary discussions, and Joseph O'Rourke for
pointing me to~\cite{randomhypergraph}.

\bibliographystyle{alpha}

\bibliography{setquery}

\appendix

\section{Negative Dependence}\label{app:negative}

Negative dependence is a fairly common property in balls-and-bins
types of problems, and can often cleanly be analyzed using the
framework of \emph{negative association} (\cite{Dubhashi96ballsand,Dubhashi96negativedependence,joagdevproschan}).

\begin{definition}[Negative Association]
  Let $(X_1, \dotsc, X_n)$ be a vector of random variables.  Then
  $(X_1, \dotsc, X_n)$ are \emph{negatively associated} if for every
  two disjoint index sets, $I, J \subseteq [n]$,
  \begin{align*}
    &\E[f(X_i, i \in I)g(X_j, j \in J)] \\\leq& \E[f(X_i, i \in I)]E[g(X_j, j \in J)]
  \end{align*}
  for all functions $f\colon \R^{\abs{I}} \to \R$ and $g \colon
  \R^{\abs{J}} \to \R$ that are both non-decreasing or both
  non-increasing.
\end{definition}

If random variables are negatively associated then one can apply most
standard concentration of measure arguments, such as Chebyshev's
inequality and the Chernoff bound.  This means it is a fairly strong
property, which makes it hard to prove directly.  What makes it so
useful is that it remains true under two composition rules:

\begin{lemma}[\cite{Dubhashi96ballsand}, Proposition 7]\label{negativeprop}
  ~
  \begin{enumerate}
  \item If $(X_1, \dotsc, X_n)$ and $(Y_1, \dotsc, Y_m)$ are each
    negatively associated and mutually independent, then $(X_1,
    \dotsc, X_n, Y_1, \dotsc, Y_m)$ is negatively associated.
  \item Suppose $(X_1, \dotsc, X_n)$ is negatively associated.  Let
    $I_1, \dotsc, I_k \subseteq [n]$ be disjoint index sets, for some
    positive integer $k$.  For $j \in [k]$, let $h_j \colon
    \R^{\abs{I_j}} \to \R$ be functions that are all non-decreasing or
    all non-increasing, and define $Y_j = h_j(X_i, i \in I_j)$.  Then
    $(Y_1, \dotsc, Y_k)$ is also negatively associated.
  \end{enumerate}
\end{lemma}

Lemma~\ref{negativeprop} allows us to relatively easily show that one
component of our error (the point error) is negatively associated
without performing any computation.  Unfortunately, the other
component of our error (the component size) is not easily built up by
repeated applications of Lemma~\ref{negativeprop}\footnote{This paper
  considers the component size of each hyperedge, which clearly is not
  negatively associated: if one hyperedge is in a component of size
  $k$ than so is every other hyperedge.  But one can consider variants
  that just consider the distribution of component sizes, which seems
  plausibly negatively associated.  However, this is hard to prove.}.
Therefore we show something much weaker for this error, namely
\emph{approximate negative correlation}:
\[
\E[X_iX_j] - \E[X_i]E[X_j] \leq \frac{1}{k^{\Omega(1)}} \E[X_i]\E[X_j]
\]
for all $i \neq j$.  This is still strong enough to use Chebyshev's
inequality.

\section{Set Query in the $\ell_1$ norm}\label{app:l1}

This section works through all the changes to prove the set query
algorithm works in the $\ell_1$ norm with $w = O(\frac{1}{\eps}k)$
measurements.

We use Lemma~\ref{thm:pointdeviation} to get an $\ell_1$ analog of
Corollary~\ref{thm:totalerror}:

\begin{align}
  \norm{1}{x'-x_S} &= \sum_{i \in S} \abs{(x'-x_S)_i} \\
  \notag &\leq \sum_{i \in S} 2\sum_{j \in S} C_{i,j} \abs{Y_j} = 2\sum_{i \in S} D_i\abs{Y_i}.
\end{align}

Then we bound the expectation, variance, and covariance of $D_i$ and
$\abs{Y_i}$.  The bound on $D_i$ works the same as in
Section~\ref{sec:componentsize}: $\E[D_i] = O(1)$, $\E[D_i^2] =
O(1)$, $\E[D_iD_j] - \E[D_i]^2 \leq O(\log^4 k / \sqrt{k})$.

The bound on $\abs{Y_i}$ is slightly different.  We define
\[
U_q' = \abs{\nu_q} + \sum_{i \in [n] \setminus S} \abs{x_i}B_{i,q}
\]
and observe that $U_q' \geq \abs{V_q}$, and $U_q'$ is NA.  Hence
\[
Z_i' = \median_{q\in L_i} U_q'
\]
is NA, and $\abs{Y_i} \leq Z_i'$.  Define
\begin{align*}
  \mu &= \E[U_q'] = \frac{d}{w}\norm{1}{x-x_S} + \frac{1}{w}\norm{1}{\nu} \\
  &\leq \frac{\eps}{k}(\norm{1}{x-x_S} + \norm{1}{\nu})
\end{align*}
then
\[
\Pr[Z_i' \geq c\mu] \leq 2^{\abs{L_i}}(\frac{1}{c})^{\abs{L_i}/2} \leq \left(\frac{4}{c}\right)^{d-2}
\]
so $\E[Z_i'] = O(\mu)$ and $\E[Z_i'^2] = O(\mu^2)$.

Now we will show the analog of Section~\ref{sec:wrapping}.  We know
\[
\norm{2}{x'-x_S} \leq 2\sum_i D_iZ_i'
\]
and
\[
\E[2\sum_i D_iZ_i'] = 2 \sum_i \E[D_i]\E[Z_i'] = k\mu'
\]
for some $\mu' = O(\frac{\eps}{k}(\norm{1}{x-x_S} + \norm{1}{\nu}))$.
Then
\begin{align*}
  \E[(\sum D_iZ_i')^2] &= \sum_i\E[D_i^2]\E[Z_i'^2] + \sum_{i \neq j} \E[D_iD_j]\E[Z_i'Z_j']\\
  &\leq \sum_i O(\mu'^2) + \sum_{i \neq j} (\E[D_i]^2 + O(\log^4 k / \sqrt{k}))\E[Z_i']^2\\
  &= O(\mu'^2k\sqrt{k}\log^4 k) + k(k-1)\E[D_iZ_i']^2\\
  \mbox{Var}(2\sum_i Z_i'D_i) &\leq O(\mu'^2k\sqrt{k}\log^4 k).
\end{align*}
By Chebyshev's inequality, we get
\[
\Pr[\norm{1}{x'-x_S} \geq (1+\alpha)k\mu'] \leq O(\frac{\log^4 k}{\alpha^2\sqrt{k}})
\]
and the main theorem (for constant $c = 1/3$) follows.  The parallel
repetition method of Section~\ref{sec:wrapping} works the same as in
the $\ell_2$ case to support arbitrary $c$.

\section{Block Heavy Hitters}\label{app:strengthenheavyhitters}

\newtheorem*{oldlemma5}{Lemma \ref{thm:block-heavy-hitters}}
\begin{oldlemma5}
  For any $b$ and $k$, there exists a family of matrices $A$ with
  $O(\frac{k}{\eps^5b}\log n)$ rows and column sparsity $O(\frac{1}{\eps^2}\log n)$
  such that we can recover a support $S$ from $Ax$ in $O(\frac{n}{\eps^2
    b}\log n)$ time with
  \[
  \norm{2}{x-x_S} \leq (1+\eps)\err{2}{k,b}{x}
  \]
  with probability at least $1 - n^{-\Omega(1)}$.
\end{oldlemma5}

\begin{proof}
  This proof follows the method of~\cite{block-heavy-hitters}, but
  applies to the $\ell_2$ norm and is in the (slightly stronger)
  sparse recovery framework rather than the heavy hitters framework.
  The idea is to perform dimensionality reduction, then use an
  argument similar to those for Count-Sketch (first in~\cite{CM06},
  but we follow more closely the description in~\cite{GI}).

  Define $s = k/b$ and $t = n/b$, and decompose $[n]$ into equal sized
  blocks $T_1, \dotsc, T_t$.  Let $x_{(T_i)} \in \R^b$ denote the
  restriction of $x_{T_i}$ to the coordinates $T_i$.  Let $U \subseteq
  [t]$ have $\abs{U}=s$ and contain the $s$ largest blocks in $x$, so
  $\err{2}{k,b}{x} = \norm{2}{\sum_{i \notin U} x_{T_i}}$.

  Choose an i.i.d. standard Gaussian matrix $\rho \in \R^{m\times b}$
  for $m = O(\frac{1}{\eps^2} \log n)$.  Define $y_{q,i} = (\rho
  x_{(T_q)})_i$, so as a distribution over $\rho$, $y_{q,i}$ is a
  Gaussian with variance $\norm{2}{x_{(T_q)}}^2$.

  Let $h_1, \dotsc, h_m\colon [t] \to [l]$ be pairwise independent
  hash functions for some $l = O(\frac{1}{\eps^3}s)$, and $g_1,
  \dotsc, g_m\colon [t] \to \{-1,1\}$ also be pairwise independent.
  Then we make $m$ hash tables $H^{(1)}, \dotsc, H^{(m)}$ of size $l$
  each, and say that the value of the $j$th cell in the $i$th hash
  table $H^{(i)}$ is given by
  \[
  H^{(i)}_j = \sum_{q: h_i(q) = j} g_i(q)y_{q,i}
  \]
  Then the $H^{(i)}_j$ form a linear sketch of $ml = O(\frac{k}{\eps^5
    b}\log n)$ cells.  We use this sketch to estimate the mass of each
  block, and output the blocks that we estimate to have the highest
  mass.  Our estimator for $\norm{2}{x_{T_i}}$ is
  \[
  z_i' = \alpha\median_{j \in [m]} \abs{H^{(j)}_{h_j(i)}}
  \]
  for some constant scaling factor $\alpha \approx 1.48$.  Since we
  only care which blocks have the largest magnitude, we don't actually
  need to use $\alpha$.

  We first claim that for each $i$ and $j$ with probability $1 -
  O(\eps)$, $(H^{(j)}_{h_j(i)} - y_{i,j})^2 \leq O(\frac{\eps^2}{s}
  (\err{2}{k,b}{x})^2)$.  To prove it, note that the probability any
  $q \in U$ with $q \neq i$ having $h_j(q) = h_j(i)$ is at most
  $\frac{s}{l} \leq \eps^3$.  If such a collision with a heavy hitter
  does not happen, then
  \begin{align*}
    \E[(H^{(j)}_{h_j(i)} - y_{i,j})^2] &= \E[\sum_{p \neq i, h_j(p) = h_j(i)} y_{p,j}^2]\\
    &\leq \sum_{p \notin U} \frac{1}{l}\E[y_{p,j}^2]\\
    &= \frac{1}{l}\sum_{p \notin U} \norm{2}{x_{T_p}}^2\\
    &= \frac{1}{l}(\err{2}{k,b}{x})^2
  \end{align*}
  By Markov's inequality and the union bound, we have
  \[
  \Pr[(H^{(j)}_{h_j(i)} - y_{i,j})^2 \geq \frac{\eps^2}{s}(\err{2}{k,b}{x})^2] \leq \eps + \eps^3 = O(\eps)
  \]

  Let $B_{i,j}$ be the event that $(H^{(j)}_{h_j(i)} - y_{i,j})^2 >
  O(\frac{\eps^2}{s} (\err{2}{k,b}{x})^2)$, so $\Pr[B_{i,j}] =
  O(\eps)$.  This is independent for each $j$, so by the Chernoff
  bound $\sum_{j=1}^m B_{i,j} \leq O(\eps m)$ with high probability in
  $n$.

  Now, $\abs{y_{i,j}}$ is distributed according to the positive half
  of a Gaussian, so there is some constant $\alpha \approx 1.48$ such
  that $\alpha \abs{y_{i,j}}$ is an unbiased estimator for
  $\norm{2}{x_{T_i}}$.  For any $C\geq 1$ and some $\delta =
  O(C\eps)$, we expect less than $\frac{1-C\eps}{2}m$ of the
  $\alpha\abs{y_{i,j}}$ to be below $(1-\delta)\norm{2}{x_{T_i}}$,
  less than $\frac{1-C\eps}{2}m$ to be above
  $(1+\delta)\norm{2}{x_{T_i}}$, and more than $C\eps m$ to be in
  between.  Because $m \geq \Omega(\frac{1}{\eps^2}\log n)$, the
  Chernoff bound shows that with high probability the actual number of
  $\alpha\abs{y_{i,j}}$ in each interval is within
  $\frac{\eps}{2} m = O(\frac{1}{\eps}\log n)$ of its expectation.  Hence
  \[
  \abs{\norm{2}{x_{T_i}} - \alpha \median_{j\in [m]} \abs{y_{i,j}}} \leq \delta\norm{2}{x_{T_i}} = O(C\eps) \norm{2}{x_{T_i}}.
  \]
  even if $\frac{(C-1)\eps}{2} m$ of the $y_{i,j}$ were adversarially
  modified.  We can think of the events $B_{i,j}$ as being such
  adversarial modifications.  We find that
  \begin{align*}
    \abs{\norm{2}{x_{T_i}} - z_i} &= \abs{\norm{2}{x_{T_i}} - \alpha \median_{j\in [m]} \abs{H_{h_j(i)}^{(j)}}} \\
    &\leq O(\eps) \norm{2}{x_{T_i}} + O(\frac{\eps}{\sqrt{s}}\err{2}{k,b}{x}).
  \end{align*}

  \[
  (\norm{2}{x_{T_i}} - z_i)^2 \leq O(\eps^2\norm{2}{x_{T_i}}^2 + \frac{\eps^2}{s}(\err{2}{k,b}{x})^2)
  \]

  Define $w_i = \norm{2}{x_{T_i}}$, $\mu = \err{2}{k,b}{x}$, and
  $\hat{U}\subseteq [t]$ to contain the $s$ largest coordinates in
  $z$.  Since $z$ is computed from the sketch, the recovery algorithm
  can compute $\hat{U}$.  The output of our algorithm will be the
  blocks corresponding to $\hat{U}$.

  We know $\mu^2 = \sum_{i \notin U} w_i^2 = \norm{2}{w_{[t]\setminus U}}^2$ and $ \abs{w_i - z_i} \leq
  O(\eps w_i + \frac{\eps}{\sqrt{s}}\mu)$ for all $i$.  We will show
  that
  \[
  \norm{2}{w_{[t] \setminus \hat{U}}}^2 \leq (1+O(\eps))\mu^2.
  \]
  This is analogous to the proof of Count-Sketch, or to
  Corollary~\ref{cor:linfthresholding}.  Note that
  \begin{align*}
    \norm{2}{w_{[t] \setminus \hat{U}}}^2 &= \norm{2}{w_{U \setminus \hat{U}}}^2 + \norm{2}{w_{[t]\setminus(U\cup \hat{U})}}^2
  \end{align*}
  For any $i \in U \setminus \hat{U}$ and $j \in \hat{U} \setminus U$, we have $z_j >
  z_i$, so
  \[
  w_i - w_j \leq O(\frac{\eps}{\sqrt{s}}\mu + \eps w_i)
  \]
  Let $a = \max_{i\in U \setminus \hat{U}} w_i$ and $b = \min_{j \in \hat{U} \setminus
    U} w_j$.  Then $a \leq b + O(\frac{\eps}{\sqrt{s}}\mu + \eps a)$,
  and dividing by $(1-O(\eps))$ we get $a \leq b(1+O(\eps)) +
  O(\frac{\eps}{\sqrt{s}}\mu)$.  Furthermore
  $\norm{2}{w_{\hat{U}\setminus U}}^2 \geq b^2\abs{\hat{U}\setminus U}$, so
  \begin{align*}
    \norm{2}{w_{U \setminus \hat{U}}}^2 \leq& \left(\norm{2}{w_{\hat{U}\setminus U}} \frac{1+O(\eps)}{\sqrt{\abs{\hat{U}\setminus U}}} + O(\frac{\eps}{\sqrt{s}}\mu)\right)^2\abs{\hat{U}\setminus U}\\
    \leq& \left(\norm{2}{w_{\hat{U}\setminus U}} (1+O(\eps)) + O(\eps\mu)\right)^2\\
    =& \norm{2}{w_{\hat{U}\setminus U}}^2(1+O(\eps)) + (2+O(\eps))\norm{2}{w_{\hat{U}\setminus U}}O(\eps\mu) \\& + O(\eps^2\mu^2)\\
    \leq& \norm{2}{w_{\hat{U}\setminus U}}^2 + O(\eps\mu^2)
  \end{align*}
  because $\norm{2}{w_{\hat{U}\setminus U}} \leq \mu$.  Thus
  \begin{align*}
\norm{2}{w - w_{\hat{U}}} = \norm{2}{w_{[t] \setminus \hat{U}}}^2
&\leq O(\eps \mu^2) + \norm{2}{w_{\hat{U}\setminus U}}^2 +
\norm{2}{w_{[t]\setminus(U\cup \hat{U})}}^2\\ 
&= O(\eps \mu^2) + \mu^2 = (1+O(\eps))\mu^2.
  \end{align*}
  This is exactly what we want.  If $S = \bigcup_{i\in \hat{U}} T_i$
  contains the blocks corresponding to $\hat{U}$, then
  \[
  \norm{2}{x - x_S} = \norm{2}{w - w_{\hat{U}}} \leq (1+O(\eps))\mu = (1+O(\eps))\err{2}{k,b}{x}
  \]
  Rescale $\eps$ to change $1+O(\eps)$ into $1+\eps$ and we're done.
\end{proof}

\end{document}